\documentclass[pdflatex,sn-mathphys-ay,iicol]{sn-jnl}

\usepackage{graphicx}%
\usepackage{multirow}%
\usepackage{amsmath,amssymb,amsfonts}%
\usepackage{amsthm}%
\usepackage{mathrsfs}%
\usepackage[title]{appendix}%
\usepackage{xcolor}%
\usepackage{textcomp}%
\usepackage{manyfoot}%
\usepackage{booktabs}%
\usepackage{algorithm}%
\usepackage{algorithmicx}%
\usepackage{algpseudocode}%
\usepackage{listings}%
\usepackage{upgreek}
\usepackage{dcolumn,adjustbox,threeparttable,siunitx,float,graphics}
\usepackage{rotating,pdflscape,color,caption,subcaption,etoolbox,comment,enumitem}
\newcolumntype{d}[1]{D{.}{.}{#1}}
\newcommand{\GG}[1]{}
\newcommand\mc[1]{\multicolumn{1}{c}{#1}}

\newcommand{\FBetaAlpha}{(\mathbf{f},\boldsymbol{\beta},\alpha)}
\newcommand{\FBetaAlphaMUL}{(\mathbf{f},\boldsymbol{\beta}^m,\boldsymbol{\alpha})}
\newcommand{\FBetaAlphaMULhat}{(\hat{\mathbf{f}},\hat{\boldsymbol{\beta}^m},\hat{\boldsymbol{\alpha}})}
\newcommand{\FBetaAlphahat}{(\hat{\mathbf{f}},\hat{\boldsymbol{\beta}},\hat{\alpha})}

\theoremstyle{thmstyleone}%
\newtheorem{theorem}{Theorem}
\theoremstyle{thmstyletwo}%

\theoremstyle{thmstylethree}%

\begin{document}

\title[Article Title]{Generalized dynamic functional principal component analysis}


\author[1,2]{\fnm{Tzung-Hsuen} \sur{Khoo}}\email{tzung-hsuen.khoo@univ-lille.fr}

\author[3]{\fnm{Issa-Mbernard} \sur{Dabo}}\email{issa.mbenard.dabo@nyu.edu}
\equalcont{These authors contributed equally to this work.}

\author*[1,4,5]{\fnm{Dharini} \sur{Pathmanathan}}\email{dharini@um.edu.my}
\equalcont{These authors contributed equally to this work.}

\author[2]{\fnm{Sophie} \sur{Dabo-Niang}}\email{sophie.dabo@univ-lille.fr}
\equalcont{These authors contributed equally to this work.}

\affil[1]{\orgdiv{Institute of Mathematical Sciences, Faculty of Science}, 
\orgname{Universiti Malaya}, 
\orgaddress{\city{Kuala Lumpur}, \postcode{50603}, \country{Malaysia}}}

\affil[2]{\orgdiv{UMR 8524 -- Laboratoire Paul Painlev\'{e}}, 
\orgname{CNRS, Inria -- Datavers, Universit\'{e} de Lille}, 
\orgaddress{\city{Lille}, \postcode{59000}, \country{France}}}

\affil[3]{\orgdiv{Mathematics, Division of Science}, 
\orgname{New York University Abu Dhabi}, 
\orgaddress{\country{United Arab Emirates}}}

\affil[4]{\orgdiv{Universiti Malaya Centre for Data Analytics}, 
\orgname{Universiti Malaya}, 
\orgaddress{\city{Kuala Lumpur}, \postcode{50603}, \country{Malaysia}}}

\affil[5]{\orgdiv{Center of Research for Statistical Modelling and Methodology, Faculty of Science}, 
\orgname{Universiti Malaya}, 
\orgaddress{\city{Kuala Lumpur}, \postcode{50603}, \country{Malaysia}}}


\abstract{
In this paper, we explore dimension reduction for functional time series. We propose a generalized dynamic functional principal component analysis (GDFPCA)
which does not rely on spectral density estimation and demonstrates strong empirical performance for both stationary and nonstationary functional time series.
We define the generalized dynamic functional principal components (GDFPCs) as static factor time series in a functional dynamic factor model and obtain their multivariate representation from a truncation of the functional dynamic factor model. Estimation is based on a least-squares reconstruction criterion and implemented via a two-step procedure for the coefficient vectors of the loading curves under a basis expansion. We establish mean-square consistency of the reconstructed functional time series under weak stationarity. Simulation studies show that GDFPCA performs comparably to dynamic functional principal component analysis (DFPCA) for stationary data, while providing improved reconstruction accuracy in nonstationary settings, where both DFPCA and functional principal component analysis (FPCA) deteriorate. Applications to real datasets support the empirical advantages observed in the simulations.
}




\maketitle

\section{Introduction}\label{sec1}
Dynamical dimensionality reduction is central to analyzing high-dimensional multivariate time series, where nonstationarity, structural breaks, and evolving cross-sectional dependence are common. The growing availability of high-dimensional time series across disciplines, including economics and finance, environmental and climate monitoring, energy systems, and public health, highlights the need for parsimonious representations for modeling, visualization, and forecasting. Data-driven dimension reduction is particularly attractive in this setting, as it yields low-dimensional representations that adapt to data complexity without imposing restrictive parametric structure. Across many application areas, however, dimension reduction is challenged by the prevalence of non-stationarity and structural change. This motivates flexible procedures that capture dynamic features such as time-varying fluctuations and uncertainty while providing accurate low-dimensional reconstructions. To address these challenges, \cite{pena2016generalized} proposed generalized dynamic principal component analysis (GDPCA), a dimension reduction approach in which components are obtained by minimizing the loss of reconstruction, 
with the aim of 
providing a dimension reduction-based representation that does not rely on spectral density estimation of the series. Building on this idea, we develop a functional extension of GDPCA in which the dynamic principal components are defined through a curve level reconstruction. This functional representation preserves the natural ordering of the measurements along a continuum and incorporates their smooth structure, allowing reconstruction accuracy to be assessed at the level of the entire curves.

Functional data analysis (FDA) provides a framework for modeling high-dimensional datasets in which discrete temporal or spatio-temporal measurements are treated as continuous functions. Within this framework, each complete curve is viewed as a single observation \citep{Ramsay2005-RAMFDA-3} of a random variable taking values in a possibly infinite-dimensional functional space. To address this inherent complexity, functional principal component analysis (FPCA) plays a central role in dimensionality reduction by projecting these infinite-dimensional objects onto a finite-dimensional subspace that captures the dominant modes of variation.

Traditional FPCA is designed for independent functional observations; however, many functional datasets in several areas (e.g, economics) display dependence across curves, such as yield curves \citep{diebold2006forecasting, hays2012functional} and intraday price curves \citep{kokoszka2015functional}.
Furthermore, observations of such dependent data are viewed as realizations of functional time series  \citep{horvath2012inference}. In such cases, ignoring the dependence between curves can affect the results of FPCA \citep{hormann2010}. 

An expanding literature has focused on the modeling of dependent, second-order stationary functional data. The inter-curve dynamics inherent in such dependent and stationary processes are characterized by lag-specific autocovariance operators. 
Extensive works were developed for estimating the singular value decompositions of the lag-covariance function for linear stationary functional autoregressive (AR) processes \citep{mas2002,bosq2002estimation}, notably the stationary functional AR(1) process \citep{bosq2000linear}. In recent years, the study of functional time series shifted from linear processes toward more general frameworks which allow for different dependence structures. \cite{hormann2010} introduced a moment-based notion of dependence for weakly dependent functional data. \cite{horvathfts2012} estimated the long-run covariance operator which is a weighted sum of the autocovariance operators. Complementing these time-domain approaches, frequency-domain methods were subsequently developed to capture the global dynamics of functional processes across all lags and frequencies. \cite{panaretos:tavakoli:2013AOS} consider the spectral density operators of a functional time series, its Fourier transform of autocovariance operator which is analogous to the spectral density matrix in the multivariate case \citep{brillinger2001time}. Building on the concept of the spectral density operator, dynamic PCA pioneered by \cite{brillinger2001time} was extended to functional data in two closely timed contributions, which are referred to as harmonic PCA \citep{panaretos:tavakoli:2013SPA} and dynamic FPCA (DFPCA) \citep{hormann2015dynamic} respectively. \cite{panaretos:tavakoli:2013SPA} introduced a Cramér–Karhunen–Loève representation for stationary functional time series and defined harmonic PCA as a rank-$K$ truncation of this representation. 

On the other hand, \cite{hormann2015dynamic} established a functional extension of Brillinger's dynamic PCA via the construction of a functional filter of $\{X_t\}$ where the filter coefficients are bounded linear operators derived from the spectral decomposition of the spectral density operator. 
More recently, \cite{vandelft_dette_2024} developed a general framework for nonstationary function-valued processes that works with time-varying spectral density operators, providing a locally stationary frequency-domain perspective for assessing structural assumptions.
Recent works in forecasting functional time series that applied the functional extension of the dynamic PCA include \cite{gao2019high,yang2022forecasting,shang2022dynamic,martinez2023surface}.

Another approach discussed to reduce the dimension of functional time series is the dynamic factor model \citep{hays2012functional,liebl2013modeling, kokoszka2015functional,martinez2022nonparametric}. In the multivariate case, dynamic factor models use a smaller number of latent variables (factors) to explain a larger fraction of the variance in multiple time series \citep{ forni2000generalized}. 
Some important work on dynamic factor models in stationary multivariate time series includes \cite{bai2002determining, diebold2006forecasting, forni2005generalized,giannone2006vars, lam2011estimation,forni2017dynamic}. In the context of functional data, \cite{hays2012functional} introduced a functional dynamic factor model (FDFM) with an application to forecast the yield curve and estimate their model using the expectation–maximization (EM) algorithm. \cite{liebl2013modeling} proposed an alternative way to estimate the FDFM by specifying the factor loading curves as the eigenfunctions of the series’ covariance operator and taking the associated scores as the factor process. \cite{kokoszka2015functional} estimated the factor process with the least squares estimator assuming that the factor loading curves are known. \cite{martinez2022nonparametric} proposed nonparametric estimators for FDFM in stationary and nonstationary processes. Another approach to dimension reduction for functional time series was recently proposed in \cite{Chang2025}, where high-dimensional functional time series are modeled through a two-step procedure transforming the data into lower-dimensional functional subseries that are mutually uncorrelated.

 In the case of non-stationary multivariate time series, \cite{stock1988testing} investigated factors in co-integrated time series. \cite{bai2004panic} proposed a dynamic factor model for large dimensional and non-stationary panel data. \cite{pena2006nonstationary} proposed a dynamic factor model for vector time series. In work more directly connected to dynamic factor models, \cite{pena2016generalized} introduced a generalized dynamic principal component analysis (GDPCA), which is of central interest in this paper. It is a dimensionality reduction approach entirely based on data analytics. \cite{smucler2019consistency} proved that when data follow a dynamic factor model with one dynamic factor, the first generalized dynamic principal component (GDPC) converges in mean square to the common part of the factor model. 

Unlike dynamic principal component analysis \citep{brillinger2001time} and its functional variants, which seek linear combinations that maximize variance over time via the spectral density, as well as dynamic factor models based on probabilistic specifications with latent dynamic factors, GDPCA \citep{pena2016generalized} defines dynamic components as filtered linear combinations without assuming a latent model. This avoids  approximating spectral measures or assuming hidden factors and instead specifies the components based on how well they reconstruct the observed time series. 
In this paper, we develop a functional extension of this approach, termed generalized dynamic functional principal component analysis (GDFPCA), and investigate its theoretical properties and empirical behavior through simulation studies and real-data applications.

The remainder of the paper is organized as follows. In section 2, we describe the methodology of the proposed GDFPCA, as well as its computation and consistency properties. In section 3, we show the results of Monte-Carlo studies with different data-generating processes that compare GDFPCA with FPCA and DFPCA \citep{hormann2015dynamic}. In section 4, we apply GDFPCA on real data and compare the approximations with FPCA and DFPCA. Section 5 concludes with recommendations for future work.

\section{Methodology}

Consider $n$  observations $X_1,...,X_n$ of a functional time-series  $({X}_t, t\in \mathbb{Z})$, valued in $\mathcal{H} =\mathcal{L}^2(\mathcal{T})$, the space of complex square-integrable functions ($\int_{\mathcal{T}} f(u) \bar{f}(u) du$,  $ \bar{f}$ is the conjugate of $f$) on $\mathcal{T}$ (a compact set in $\mathbb{R}$), with finite (Lebesgue-) measure.
In practice, most applications are real-valued.\\
Let the inner product $ \langle \cdot, \cdot \rangle: \mathcal{H} \times \mathcal{H} \to \mathbb{R} $, for $f, g \in \mathcal{H}$:
	
	\begin{equation*}
		\langle f, g \rangle= \int_{\mathcal{T}} f(u) \bar{g}(u) du.
		\end{equation*}
	Then, $\mathcal{H}$ is a Hilbert space with respect to the scalar product $\langle\cdot, \cdot \rangle$ with $\|f\|_\mathcal{H}= \langle f,f \rangle^{1/2}$ defining a norm.

A functional variable $X \in \mathcal{H}$ possesses a covariance operator\\ $C := \mathbb{E}[(X-\mu)\otimes(X-\mu)]$ (where $\mu$ is the mean curve define by $\mu(t)=\mathbb{E}X(t)$ with $t \in  \mathcal{T} $) with kernel $c(u,v)=cov(X(u),X(v))$ ($u,v \in  \mathcal{T}$). Then,  the integral operator $C$ is defined by 
\begin{equation*}
    (C f)(u) = \int_{\mathcal{T}} c(u,v) f(v) dv, \quad f \in\mathcal{L}^2(\mathcal{T}),~ u \in \mathcal{T}.
\end{equation*}
Suppose that each  $({X}_t, t\in \mathbb{Z})$ is a weakly stationary functional variable defined on some probability space $(\Omega,\mathcal{A}, P)$:\\
(i) $\mathbb{E}(X_t(u))=\mathbb{E}(X_0(u))=\mu(u)$, $u \in \mathcal{T}$, \\
(ii) for all $h \in \mathbb{Z}$, and $u, v \in  \mathcal{T}$; $c_{{h}}(u,v):=Cov\bigg(X_h(u),X_{0}(v)\bigg)=Cov\bigg(X_{t+{h}}(u),X_t(v)\bigg)$.\\
Let $C_h$, $h\in \mathbb{Z}$ be the covariance operator corresponding to $c_h$, note that $C_0=C$.
The operator $C$ admits an eigen-decomposition,
\begin{equation*}
    C(x) = \sum_{j=1}^\infty \lambda_j \langle x,\phi_j \rangle \phi_j,
\end{equation*}
where $\lambda_{1} \geq \lambda_{2} \geq \ldots \geq 0$ are the eigenvalues of $C$ and $\{\phi_j\}_{j\geq1}$ the associated orthonormal eigenfunctions. This yields the Karhunen-Lo\`eve representation of $X$,
\begin{equation*}
    X = \sum_{j=1}^\infty \langle X,\phi_j \rangle \phi_j.
\end{equation*}

This representation does not take into account the serial dependency on the functional time-series  observations, $X_t$ \citep{hormann2015dynamic}. Without a loss of generality let us assume that $({X}_t, t\in \mathbb{Z})$ is centered. 

    Let us assume that the time series observations $X_1,...,X_n$ follow the model below:
     \begin{align}\label{eq1}
       X_{t}(u) &=Y_t(u)+ e_{t}(u),\, t=1,...,n, \\
       Y_t(u) &=\sum_{h=0}^K f_{t-h}\beta_{h}(u),
      \end{align}
    where $\{\beta_{h}\}$ are factor loading curves, $\{f_t\}$  are scalar factor time series,
    $\{e_t\}$ is a sequence of centered independent and identically distributed (i.i.d) functional random variables in $\mathcal{L}^2(\mathcal{T})$. By assuming the orthonormality of the factor loading curves, we further assume the following: \\~\\
    \textbf{Assumptions:}  \vskip 1em
\noindent \textbf{A1.} $\{e_t\}$ is uncorrelated with $\{Y_t\}$; \\
    \textbf{A2.} Let $C_{Y,h}$ and $C_{e}$ ($\{e_t\}$ is a white noise) be the covariance operators of $\{Y_t\}$ and $\{e_t\}$ respectively. Let $\lambda_j^Y$, $\lambda_j^e$ be their respective ordered j$th$ eigenvalues. Their maximum eigenvalues $\lambda_{max}^Y$ and $\lambda_{max}^e$ are finite. \\~\\ 
Motivated by the definition of multivariate GDPC \citep{pena2016generalized}, we define the first GDFPC with $K$ lags as a vector $\mathbf{f}= (f_{1-K},...,f_{n})^\top$, such that the reconstruction of the functional variable $X_{t}(.)$ as a linear combination of $(f_{t-K},...,f_{t-1},f_{t})^\top$ is optimal with respect to the mean squared error (MSE) criterion. Additionally, we choose $\mathbf{f}$ such that $\sum_{t=1-K}^nf_t=0$ and $(1/(n+K)) \sum_{t=1-K}^nf_t^2=1$ to ensure identifiability. Specifically, given a $(n+K)$-dimensional vector $\mathbf{f}$, a $(K+1)$-dimensional vector of loading curves $\boldsymbol{\beta}=(\beta_0(u),...,\beta_K(u))^\top$ and a mean function $\alpha(u)$, the reconstruction of $X_{t}(.)$ is defined as:
\begin{align}
               X_{t}^{R,b}\FBetaAlpha &= \alpha(u) + Y_t(u) \nonumber \\
               &= \alpha(u) +\sum_{h=0}^K f_{t-h}\beta_{h}(u),
            \end{align}  
where the superscripts $R$ and $b$ stand for reconstruction and backward respectively. The MSE loss function when we reconstruct the $n$ functional time series using $\mathbf{f},\boldsymbol{\beta},{\alpha}$ is:
\begin{align}\label{Equation:MSE_H}
                \text{MSE}\FBetaAlpha= \dfrac{1}{n}\sum_{t=1}^{n}\left\|(X_{t}-X_{t}^{R,b}\FBetaAlpha)\right\|^2_\mathcal{H}.
            \end{align}
The values which minimize the above function is denoted as $\FBetaAlphahat$, is then the first GDFPC, and the second GDFPC is defined as the first GDFPC of the residuals:
\begin{align}\label{residuals}
               r_{t} = X_{t} - X_{t}^{R,b}\FBetaAlphahat, \; 1\leq t\leq n,
            \end{align}    
             and higher order GDFPC are defined similarly.
             
    \subsection{Truncated GDFPC}

By letting $\{\phi_j\}$ be an orthonormal basis of $\mathcal{H}=\mathcal{L}^2(\mathcal{T})$, we expand the functional objects in equation (\ref{eq1}):
 \begin{align}
               X_{t}(u) = \sum_{j=1}^\infty\chi_{t,j}\phi_j(u);\, \\
               \beta_{h}(u) = \sum_{j=1}^\infty\beta_{h,j}\phi_j(u);\, \\  
               e_{t}(u) = \sum_{j=1}^\infty\varepsilon_{t,j}\phi_j(u).
            \end{align} 
  Then,  equation (\ref{eq1}) may be rewritten as:
\begin{align}
    \sum_{j=1}^\infty \chi_{t,j}  \phi_j(u) &=Y_t(u)+\sum_{j=1}^\infty\varepsilon_{t,j}\phi_j(u),\, t=1,...,n, \\
       Y_t(u) &=\sum_{h=0}^K \sum_{j=1}^\infty f_{t-h}\beta_{h,j}\phi_j(u).
            \end{align}
Let us consider a truncation of the previous equation,
\begin{align}\label{approx_X_t}
   X_t(u) \approx \sum_{j=1}^m \chi_{t,j}  \phi_j(u) &=\sum_{h=0}^K \sum_{j=1}^m f_{t-h}\beta_{h,j}\phi_j(u) \nonumber\\&+\sum_{j=1}^m\varepsilon_{t,j}\phi_j(u),\, t=1,...,n.
            \end{align}
      We assume that the same truncated dimension $m=m_n=o(n)$ across $h$, $m$ goes to $\infty$ with $n$.
We derive 
the truncated reconstruction
\begin{align}
               X_{t}^{R,b}\FBetaAlpha &\approx  \sum_{j=1}^m \alpha_j \phi_j(u) \nonumber\\ &+ \sum_{h=0}^K \sum_{j=1}^m f_{t-h}\beta_{h,j}\phi_j(u),
            \end{align}      
and, from (\ref{approx_X_t}), the dynamical factor model on the scores $\chi_{t,j} $
\begin{align}
             \chi_{t,j} = \sum_{h=0}^K\beta_{h,j}f_{t-h} +\varepsilon_{t,j},\; t=1,...,n,\; j=1,...,m.
            \end{align}       

The parameters $\mathbf{f}$, $\boldsymbol{\beta}^m := (\beta_{h,j}) \in \mathbb{R}^{(K+1) \times m}$, $\boldsymbol{\alpha}=(\alpha_{j}) \in \mathbb{R}^{ m}$ are then estimated by minimizing the following least squares function:
\begin{align}\label{Equation:MSE} 
                &\text{MSE}\FBetaAlphaMUL \nonumber\\
                &=\dfrac{1}{nm}\sum_{t=1}^{n}\sum_{j=1}^{m}\left(\chi_{t,j} -\alpha_j -\sum_{h=0}^K f_{t-h}\beta_{h,j}\right)^2, \nonumber\\
                &= \dfrac{1}{nm}\sum_{t=1}^{n}\left\|\boldsymbol{\chi}_t - \boldsymbol{\chi}_t^{R,b}\FBetaAlphaMUL\right\|^2,        
            \end{align}\\
where $\boldsymbol{\chi}_t^{R,b}\FBetaAlphaMUL= (\alpha_j+\sum_{h=0}^K  f_{t-h} \beta_{h,j})_{j=1}^m \in \mathbb{R}^{m}$ is similar to the reconstruction in multivariate GDPC \citep{pena2016generalized} of the vector $\boldsymbol{\chi}_{t}=(\chi_{t,1},...,\chi_{t,m})^\top$. The values $\FBetaAlphaMUL$ which minimize the above least squares function is denoted as $\FBetaAlphaMULhat$ and then
\begin{align}
               \chi_{t,j}^{R,b}\FBetaAlphaMULhat = \hat\alpha_j +\sum_{h=0}^K \hat f_{t-h} \hat\beta_{h,j},\; j=1,...,m.
            \end{align}  
The minimization of equation (\ref{Equation:MSE_H}) can then be approximately solved by minimizing equation (\ref{Equation:MSE}) based on the truncation of the functional objects $X_t(.)$. Then the estimation of $\FBetaAlpha$ is $\FBetaAlphahat$ with  $\hat\alpha(u) =\sum_{j=1}^m \hat\alpha_j \phi_j(u)$, $\hat\beta_{h}(u) =\sum_{j=1}^m \hat\beta_{h,j} \phi_j(u)$. Then, the estimated reconstruction of $X_{t}(.)$ is 
\begin{align}
               X_{t}^{R,b}\FBetaAlphahat = \hat\alpha(u) +\sum_{h=0}^K \hat f_{t-h}\hat\beta_{h}(u).
            \end{align}    

\subsection{Computation of GDFPC}\label{sec:computation}      
We see that the truncated $\FBetaAlpha$ can be equivalently solved by minimizing the MSE loss function when we reconstruct the vector $\boldsymbol{\chi}_{t}=(\chi_{t,1},...,\chi_{t,m})^\top$. To compute $\FBetaAlphaMULhat$, we follow the approach of a two-step iterative algorithm in \cite{pena2016generalized}. That is,

\smallskip
\begin{enumerate}[label=\textbf{\arabic*.}, leftmargin=*]
    \item Compute the least-squares estimator of $\mathbf{f}$ given the matrix $\boldsymbol{\beta}^m$ and the vector $\boldsymbol{\alpha}$. Likewise, compute the least-squares estimators of $(\boldsymbol{\beta}^m_j, \alpha_j)$ for each $j$ given $\mathbf{f}$.

    \item Initialize the GDFPC as $\mathbf{f}^{(0)}$, and iteratively update (\ref{Equation:beta expression}) and (\ref{Equation:f expression}) until the convergence criterion (\ref{mse_criteria}) is satisfied for a prescribed tolerance $\varepsilon$.
\end{enumerate}
\smallskip

Given GDFPC $\mathbf{f}$, the coefficient matrix $\boldsymbol{\beta}^m$ and the vector $\boldsymbol{\alpha}$ can be readily computed by minimizing (\ref{Equation:MSE}). We rewrite (\ref{Equation:MSE}) as
            \begin{align}\label{Equation:MSE_modif}
                \text{MSE}\FBetaAlphaMUL
                &= \dfrac{1}{nm}\sum_{j=1}^{m}\left\|\boldsymbol{\chi}_{j} -\mathbf{F}(\mathbf{f})\tilde{\boldsymbol{\beta}}_j^m\right\|^2, 
            \end{align}
where $\boldsymbol{\chi}_{j} = (\chi_{1,j},\chi_{2,j},...,\chi_{n,j})^{\top}$, $\mathbf{F(f)}$ $\in \mathbb{R}^{n \times (K+2)}$, with the $t$-th row $(f_t,f_{t-1},...,f_{t-K}, 1)$. $\tilde{\boldsymbol{\beta}}_j^m =(\beta_{0,j},\beta_{1,j},...,\beta_{K,j}, \alpha_j)^\top$. Differentiating (\ref{Equation:MSE_modif}) with respect to $\tilde{\boldsymbol{\beta}}_j^m$ leads to
\begin{align}\label{Equation:beta expression}
               \begin{pmatrix}
                      \boldsymbol{\beta}^m_{j} \\ 
                      \alpha_{j}\\
                \end{pmatrix} = (\mathbf{F(f)}^\top \mathbf{F(f)})^{-1}\mathbf{F(f)}^\top \boldsymbol{\chi}_{j},
           \end{align}
where $\boldsymbol{\beta}^m_{j}, j=1,...,m$ are the rows of $\boldsymbol{\beta}^m$,  
On the other hand, given $\boldsymbol{\beta}^m$ and $\boldsymbol{\alpha}$, $\mathbf{f}$ can be computed by differentiating (\ref{Equation:MSE}) with respect to $f_t$, $t=1-K,...,n$. This leads to
\begin{align}\label{Equation:f_MSE}
    \sum_{j=1}^m \sum_{v=t \vee 1}^{n \wedge (t+K)}(\chi_{v,j}-\alpha_j-\sum_{h=0}^{K}f_{v-h}\beta_{h,j})\beta_{v-t,j} = 0,
\end{align}
where $a\vee b$ equals maximum of $a$ and $b$, and $a \wedge b$ equals minimum of $a$ and $b$. Rearrange (\ref{Equation:f_MSE}) such that
\begin{align}\label{Equation:f_MSE_2}
    &\sum_{j=1}^m \sum_{v=t \vee 1}^{n \wedge (t+K)}(\chi_{v,j}-\alpha_j)\beta_{v-t,j} \nonumber\\ &= \sum_{j=1}^m \sum_{v=(t \vee 1)}^{n \wedge (t+K)}\sum_{h=0}^{K}f_{v-h}\beta_{h,j}\beta_{v-t,j}.
\end{align}
Let the left hand side and right hand side of (\ref{Equation:f_MSE_2}) be $a_t(\boldsymbol{\alpha},\boldsymbol{\beta}^m)$ and $b_t(\mathbf{f},\boldsymbol{\beta}^m)$ respectively. Substituting $q = v-t$ in $a_t(\boldsymbol{\alpha},\boldsymbol{\beta}^m)$, we have
\begin{align}\label{MSE_at(alpha,beta)}
    a_t(\boldsymbol{\alpha},\boldsymbol{\beta}^m) = \sum_{j=1}^m \sum_{q=0 \vee (1-t)}^{(n-t) \wedge (K)}(\chi_{t+q,j}-\alpha_j)\beta_{q,j}.
\end{align}
\\
Let $\mathbf{a}(\boldsymbol{\alpha},\boldsymbol{\beta}^m) = (a_{1-K}(\boldsymbol{\alpha},\boldsymbol{\beta}^m), a_{2-K}(\boldsymbol{\alpha},\boldsymbol{\beta}^m),...,a_n(\boldsymbol{\alpha},\boldsymbol{\beta}^m))^\top$, and we express it as
\begin{align}\label{MSE_a(alpha,beta)}
    \mathbf{a}(\boldsymbol{\alpha},\boldsymbol{\beta}^m) = \sum_{j=1}^m \mathbf{C}_j(\alpha_j)\boldsymbol{\beta}^m_j
\end{align}
where $\mathbf{C}_{j}(\alpha_{j}) = (c_{j,t,q}(\alpha_{j}))_{1-K\leq t\leq n, 0\leq q\leq K} \in \mathbb{R}^{(n+K) \times(K+1)}$, $c_{j,t,q}(\alpha_{j}) = \chi_{t+q,j}-\alpha_j $ if $0 \vee (1-t) \leq q \leq (n-t) \wedge K$, and 0 otherwise.

Now, we substitute $q=v-h$ in $b_t(\mathbf{f},\boldsymbol{\beta}^m)$, and 
\begin{align}\label{MSE_bt(f,beta)}
b_t(\mathbf{f},\boldsymbol{\beta}^m) = \sum_{j=1}^m \sum_{v=(t \vee 1)}^{n \wedge (t+K)}\sum_{q=v-K}^{v}\beta_{v-q,j}\beta_{v-t,j}f_{q},
\end{align}\\
and let $\mathbf{D}_{j}(\boldsymbol{\beta}^m_{j})= (d_{j,t,q} (\boldsymbol{\beta}^m_{j}))_{1-K\leq t\leq n, 1-K\leq q\leq n} \in \mathbb{R}^{(n+K)\times(n+K)}$, where $d_{j,t,q}(\boldsymbol{\beta}^m_{j})= \sum_{v=(t \vee 1)}^{n \wedge (t+K)}\beta_{v-q,j}\beta_{v-t,j}$ if $t \vee 1 \leq q \leq (n-K)\wedge (t)$ and 0 otherwise.
We then have,
\begin{align}\label{MSE_b(f,beta)}
\mathbf{b}(\mathbf{f},\boldsymbol{\beta}^m) = \sum_{j=1}^m\mathbf{D}_{j}(\boldsymbol{\beta}^m_{j})\mathbf{f},
\end{align}
where $\mathbf{b}(\mathbf{f},\boldsymbol{\beta}^m) = (b_{1-K}(\mathbf{f},\boldsymbol{\beta}^m), b_{2-K}(\mathbf{f},\boldsymbol{\beta}^m),...,b_n(\mathbf{f},\boldsymbol{\beta}^m))^\top$. Let $\mathbf{D}(\boldsymbol{\beta}^m)= \sum_{j=1}^m\mathbf{D}_{j}(\boldsymbol{\beta}^m_{j})$. Combining (\ref{MSE_a(alpha,beta)}) and (\ref{MSE_b(f,beta)}), we have
\begin{align}\label{Equation:f expression}
   \mathbf{f} &= \mathbf{D}(\boldsymbol{\beta}^m)^{-1}\sum_{j=1}^{m}\mathbf{C}_{j}(\boldsymbol{\alpha})\boldsymbol{\beta}^m_{j},
\end{align}
We initialize the GDFPC by setting $\mathbf{f}^{(0)}$ such that its first $n$ entries correspond to the first ordinary principal component of $(\boldsymbol{\chi}_1, \boldsymbol{\chi}_2, \ldots, \boldsymbol{\chi}_n)^\top$, with the remaining $K$ entries set to zero. Then, $\FBetaAlphaMULhat$ can be computed with a two-step iterative process based on equations (\ref{Equation:beta expression}) and (\ref{Equation:f expression}):
\\~\\
\textbf{Step 1:} Based on (\ref{Equation:beta expression}), define $\boldsymbol{\beta}^{m(h)}_j$ and $\alpha_j^{(h)}$, for $1 \leq j \leq m$, 
\begin{align}
               \begin{pmatrix}
                      \boldsymbol{\beta}^{m(h)}_{j} \\ 
                      \alpha^{(h)}_{j}\\
                \end{pmatrix} = (\mathbf{F}(\mathbf{f}^{(h)})^\top \mathbf{F}(\mathbf{f}^{(h)}))^{-1}\mathbf{F}(\mathbf{f}^{(h)})^\top \boldsymbol{\chi}_{j}. \nonumber
           \end{align}
\textbf{Step 2:} Based on (\ref{Equation:f expression}), define $\mathbf{f}^{(h+1)}$ by,
\begin{align}
    \mathbf{f}^*&= \mathbf{D}(\boldsymbol{\beta}^{m(h)})^{-1}\sum_{j=1}^{m}\mathbf{C}_{j}(\boldsymbol{\alpha}^{(h)})\boldsymbol{\beta}^{m(h)}_{j},\label{step2_firsteq}\\
    \textbf{f}^{(h+1)} &= (n+K)^{1/2}\frac{\mathbf{f}^* -\bar{\mathbf{f}}^*}{\left\| \mathbf{f}^* -\bar{\mathbf{f}}^* \right\|}, \label{step2_secondeq} 
\end{align}
where $\bar{\mathbf{f}}^*$ is the mean of $\mathbf{f}^*$, and (\ref{step2_secondeq}) ensures the GDFPC is normalized to be identifiable. 
The two-step iterations is stopped when

\begin{align}\label{mse_criteria}
1 -\dfrac{\text{MSE}(\mathbf{f}^{(h+1)},\boldsymbol{\beta}^{m(h+1)},\boldsymbol{\alpha}^{(h+1)})}{\text{MSE}(\mathbf{f}^{(h)},\boldsymbol{\beta}^{m(h)},\boldsymbol{\alpha}^{(h)})} < \epsilon,
\end{align}
where $\epsilon$ is some value.

To calculate higher order GDFPC, the two-step iterative algorithm described above can be performed on the scores of the truncation of the residuals in (\ref{residuals}). We follow the approach of \cite{pena2016generalized} in selecting the optimal number of lags $K$. Given $K_{max}$, $k$ is chosen among $(0,...,K_{max})$ to be the value which minimizes some type of criterion. In particular, we used leave-one-out (LOO) cross-validation \citep{seber2003linear}
\begin{align}
    LOO_k =\frac{1}{nm}\sum_{j=1}^m\sum_{t=1}^n\frac{w_{jt}^2}{(1-h_{k,tt})^2}
\end{align}
where $w_{i,t}=\chi_{t,j}-\hat{\chi}_{t,j}$, $h_{k,tt}$ are the diagonal elements of the matrix $\mathbf{F(f)}(\mathbf{F(f)}^\top \mathbf{F(f)})^{-1}\mathbf{F(f)}^\top$.

\subsection{Consistency of GDFPC}
 Let $\psi_t^m=(\boldsymbol{\beta}^m)^{\top}F_{t}$.           
 It is readily shown with Theorem 1 from \cite{smucler2019consistency} that the $\chi^{R,b}_{t}\FBetaAlphaMULhat$ converges in mean square to $\psi^m_t$ and $\chi_{t}$ under the following assumptions A3 to A7 \citep{smucler2019consistency}. With this result, we also prove the mean-square convergence of $X^{R,b}_{t}\FBetaAlphahat$ to the common part of the functional factor model $Y_{t}$ and $X_{t}$. Throughout this section, we let $\chi, \hat{\chi},\mathcal{E}, \Psi, \hat{F}, F$  be the matrices with rows  $\chi_t^\top, \chi_t^{R,b}\FBetaAlphaMULhat ^\top, \varepsilon_t^\top, (\psi_t^{m})^\top, \hat{F}_t^\top$, and $F_t^\top$,  $t=1,...,n$ respectively. 
\\~\\
         \textbf{Assumptions:} \vskip 1em
         \noindent \textbf{A3.} 
          $\psi^m_{t}$ and $\varepsilon_{t}=(\varepsilon_{t,1},...,\varepsilon_{t,m})$ represents a zero-mean, second order $m$-dimensional stationary process that has a spectral density.  \\
          \textbf{A4.} 
          $F_{t}$ is a second order stationary process. $F_{t}$ and $\varepsilon_{t}$ are uncorrelated for all $t$; Let $\Sigma^{\chi}, \Sigma^{\varepsilon}$ be the covariance operator of $\chi_t$  and $e_t$ each with finite maximum of the eigenvalues. \\
         \textbf{A5.} 
            $\lim_{n\to\infty} \beta^m (\beta^m)^{\top}/m = S$, where $S$ is positive definite. \\
         \textbf{A6.} 
            $\limsup_{m} \sum_{u\in\mathbb{Z}}[E\{\varepsilon^{\top}_{t}\varepsilon_{t+u}\}/m]^2 < \infty$. \\
         \textbf{A7.} 
           $\limsup_{m} (1/m) \sup_{t,s\in \mathbb{Z}} \\
           \sum_{j=1}^{m}\sum_{i=1}^{m}|\text{cov}(\varepsilon_{t,i}\varepsilon_{s,i},\varepsilon_{t,j}\varepsilon_{s,j})| < \infty$.\\
         \textbf{A8.}  there exist $\zeta>1$, $\tau>\zeta/2+1$,   $|\chi_{t,j}| = O\left(j^{-\tau}\right)$.
\\~\\

Assumptions \textbf{A3}-\textbf{A7} are similar to the assumptions used by \citet{stock2002forecasting,smucler2019consistency}. Assumptions \textbf{A3}-\textbf{A4} ensure the uniqueness of the decomposition, $\chi_t = \psi^m_t + \varepsilon_t$ \citep{forni2015dynamic}. Assumption \textbf{A5} is a standard assumption in ordinary least square estimation. Assumption \textbf{A6} allows for serial correlations in $\varepsilon_t$, and assumption \textbf{A7} constrains the size of fourth moments of $\varepsilon_t$ \citep{stock2002forecasting}. Assumption \textbf{A8} is usual in FDA and may be fulfilled when using an eigen-basis and gaussian processes, orthogonal polynomials or Fourier functions and a smooth condition of the functional objects \citep{cai2006prediction}. For example, if $\beta_h$ is $\tau$-times continuously differentiable function with respect to $m$ suitable basis functions it is of the order of at most $m^{-2 \tau}$. 
\\~\\
\begin{theorem}[]\label{thm1}
Assume that A1-A7 hold. Let the common part of a dynamic factor model be defined by $\psi_t^m=(\boldsymbol{\beta}^m)^{\top}F_{t}$, where $\boldsymbol{\beta}^m := (\beta_{h,j}) \in \mathbb{R}^{(K+1) \times m}$ and $F_{t}=(f_{t},\cdots,f_{t-K})^{\top}$. Secondly, let $\chi^{R,b}_{t}\FBetaAlphaMULhat$ be the reconstruction of $\chi_{t}$ with the estimated $\FBetaAlphaMULhat$. Then, as $n \to \infty$ and $m \to \infty$, 
\begin{align}
          \dfrac{1}{nm}\sum_{t=1}^{n}\left\| \psi^m_{t}-\chi^{R,b}_{t}\FBetaAlphaMULhat \right\|^2 = O_{P}\left(\dfrac{1}{m^{1/4}} \right), \\
          \dfrac{1}{nm}\sum_{t=1}^{n}\left\| \chi_{t}-\chi^{R,b}_{t}\FBetaAlphaMULhat \right\|^2 = O_{P}\left(\dfrac{1}{m^{1/2}} \right).
        \end{align} 
\end{theorem}

In the following theorem we supposed that  $\{\phi_j\}$ is the eigenbasis functions of the operator $C(.,.)$.

\begin{theorem}[]\label{thm2}
Assume that A1-A8 hold. Let the common part of a functional dynamic factor model be defined by $Y_t=(\boldsymbol{\beta})^{\top}F_{t}$, where $\boldsymbol{\beta} := (\beta_{0}(.),...,\beta_{K}(.))$ is a vector of $K+1$ factor loading curves and $F_{t}=(f_{t},\cdots,f_{t-K})^{\top}$. Secondly, let $X^{R,b}_{t}\FBetaAlphahat$  be the reconstruction of $X_{t}$ with the estimated $\FBetaAlphahat$. Then, as $n \to \infty$ and $m =n^{\frac{1}{\zeta+2\tau}}$, 
\begin{align}
         \frac{1}{n}\sum_{t=1}^{n}\left\| X_t-X^{R,b}_{t}\FBetaAlphahat \right\|^2_\mathcal{H} = O_{P}\left( \dfrac{1}{n^{1/(2(\zeta+2\tau))}} \right).
        \end{align}
\end{theorem}
The proofs of Theorem 1 and Theorem 2 are in Appendix \ref{appendix}.

\section{Simulation Studies}

        We perform simulation studies to compare the performance of GDFPCA with that of FPCA and DFPCA \citep{hormann2015dynamic} under three different data-generating mechanisms: a first-order functional autoregressive process (FAR(1)), Wiener process, and a dynamic factor model. Our goal is to evaluate how GDFPCA performs when the data arise from a stationary model, FAR(1); a non-stationary model, Wiener process; and a dynamic factor model \citep{forni2009opening}. 
        
        We set the number of lags for DFPCA to 10 for the estimation of the spectral density operator. For GDFPCA, the lag order for the score process is selected using a data-driven procedure described in Section~\ref{sec:computation}, with the maximum lag fixed at 10. 

        For each simulated setting, PCs are computed under all three PCA frameworks, and the functional series is approximated using the corresponding components. The number of principal components (PC), $p$, is selected such that the GDFPCA approximation explains at least 80\% of the total variance of the functional time series. To ensure comparability across methods, the same value of $p$ is subsequently used for FPCA and DFPCA when approximating the series.
        Each simulation setting is simulated 100 times. We evaluate approximation accuracy using the normalized mean squared error (NMSE).

        \subsection{Simulation results for a stationary first order functional autoregressive FAR(1) process}\label{Section: Stationary FAR}
        
        We generate the FAR(1) process: $X_{n+1} = \psi(X_{n})+\epsilon_{n+1}$ 
        using the package \texttt{freqdom.fda} \citep{hormann2015dynamic}. 
        The simulation process of FAR(1) was performed in a finite dimension $d$:
        \begin{align}\label{eq.FAR}
            \langle X_{n+1},v_{j} \rangle 
              &= \langle \psi(X_{n}),v_{j} \rangle + \langle \epsilon_{n+1},v_{j} \rangle \nonumber \\  
              &= \langle \psi(\sum_{i=1}^{\infty} \langle X_{n},v_{i} \rangle v_{i}),v_{j} \rangle + \langle \epsilon_{n+1},v_{j} \rangle \nonumber \\
              &\approx \sum_{i=1}^{d} \langle X_{n},v_{i} \rangle \langle \psi(v_{i}),v_{j} \rangle + \langle \epsilon_{n+1},v_{j} \rangle,
            \end{align}        
        where $(v_{i}), i \in \mathbf{N}$ are the Fourier basis functions on [0,1]. Leaving $\mathbf{X}_{n} =(\langle X_{n}, v_{1}\rangle, ..., \langle X_{n}, v_{d}\rangle)'$ and $\boldsymbol{\epsilon}_{n} = (\langle \epsilon_{n+1}, v_{1}\rangle,...\langle \epsilon_{n+1}, v_{d}\rangle)'$, the first $d$ Fourier coefficients of $X_{n}$ can be approximated as the equation VAR(1) $\mathbf{X}_{n+1} =  \mathcal{B}\mathbf{X}_{n} + \boldsymbol{\epsilon}_{n}$, where $\mathcal{B} =(\langle \psi(v_{i}), v_{j} \rangle: 1 \leq i,j \leq d)$ \citep{hormann2015dynamic}. We use the default option of $\mathcal{B} = \kappa\mathbf{G} / 2||\mathbf{G}||$ where $\mathbf{G}_{ij} = \exp{-(i+j)}$ where $\mathbf{G}_{ij} \rightarrow 0$ as $i,j \rightarrow \infty$. 
        
        We consider the simulation settings with $d = 15, 31, 45$, $\kappa = 0.3, 0.6, 0.9$, and a sample size of $n = 300$. The values of $d$ are chosen to satisfy the requirement that the number of Fourier basis functions be odd. The choice $n = 300$ is motivated by the observation that performance remains essentially unchanged for larger sample sizes. Across all configurations, as $\kappa$ increases, a larger number of principal components is needed for the approximated series with GDFPCA to capture at least 80\% of the variance in the simulated time series. In particular, we find that $p = 6, 10, 13$ is adequate for $\kappa = 0.3, 0.6, 0.9$, respectively.
        
        Figure \ref{boxplot_NMSE_FAR1} presents boxplots of the NMSE distribution for the three PCA frameworks. For small values of $d$ and $\kappa$, GDFPCA performs comparably to DFPCA and better than FPCA. As $d$ and $\kappa$ grow larger, GDFPCA achieves superior performance relative to the other methods.

            \begin{figure}[!ht]
                 \centering
                 \includegraphics[width=0.5\textwidth]{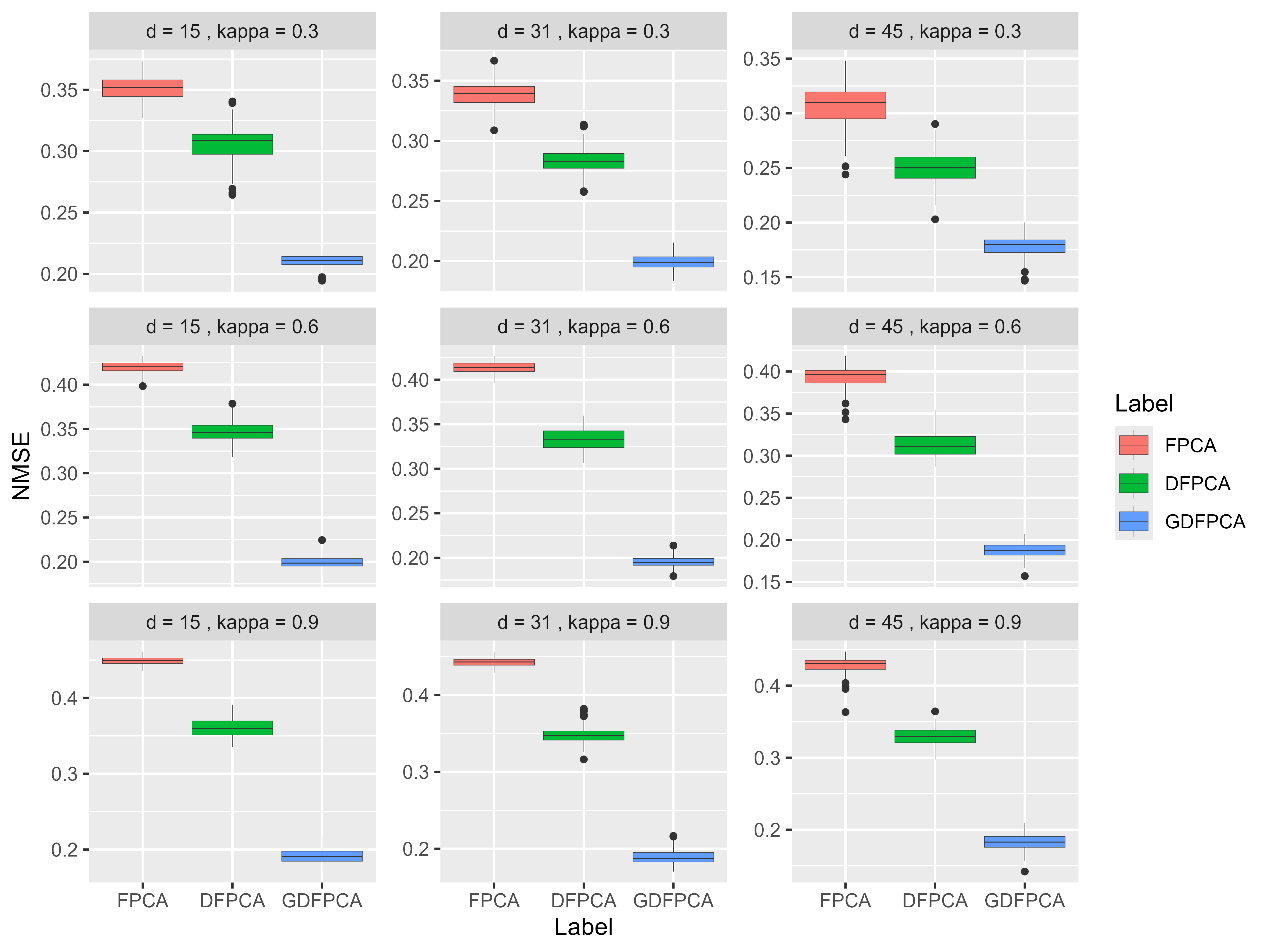} 
                 \caption{Boxplots of NMSE of 100 iterations of simulated stationary FAR(1) processes.}
                 \label{boxplot_NMSE_FAR1}
            \end{figure}            
            
\subsection{Simulation results for a non-stationary Wiener process}
         In this section, the Wiener process is considered. The functional data $X_1,X_2,\cdots,X_N$ are simulated using the truncated Karhunen-Lo\`eve expansion of the Wiener process $W_{t}$ \citep{bosq2000linear}, which is expressed as: 
                \begin{align}\label{eqn:wiener}
                     {W}_{t} &= \sqrt{2} \sum_{k=1}^{\infty}\xi_k\dfrac{\sin\left((k-\frac{1}{2})\pi t\right)}{(k-\frac{1}{2})\pi} \nonumber \\ &\approx \sqrt{2} \sum_{k=1}^{M}\xi_k\dfrac{\sin\left((k-\frac{1}{2})\pi t\right)}{(k-\frac{1}{2})\pi},
                \end{align}
where ${\xi_k}$ are independent Gaussian variables with a mean of zero with a variance of $4/{\pi^2(2k-1)^2}$, which corresponds to the $k$-th eigenvalue of the truncated expansion in (\ref{eqn:wiener}). The number of basis functions $M$ in (\ref{eqn:wiener}) is chosen to be 100. The number of observations $N=50,75,100$ and the number of observation points $T=50,100,200$ are considered.

A choice of $p=1$ is sufficient for the approximated series with GDFPCA to satisfy the 80\% explained variance threshold. The boxplots showing the NMSE distributions for the three PCA frameworks are presented in Figure \ref{boxplot_NMSE_WIENER}. Consistent with earlier findings, GDFPCA systematically achieves the lowest NMSE in all configurations, outperforming the other methods. 

                \begin{figure}[!ht]
                \centering
                    \includegraphics[width=0.5\textwidth]{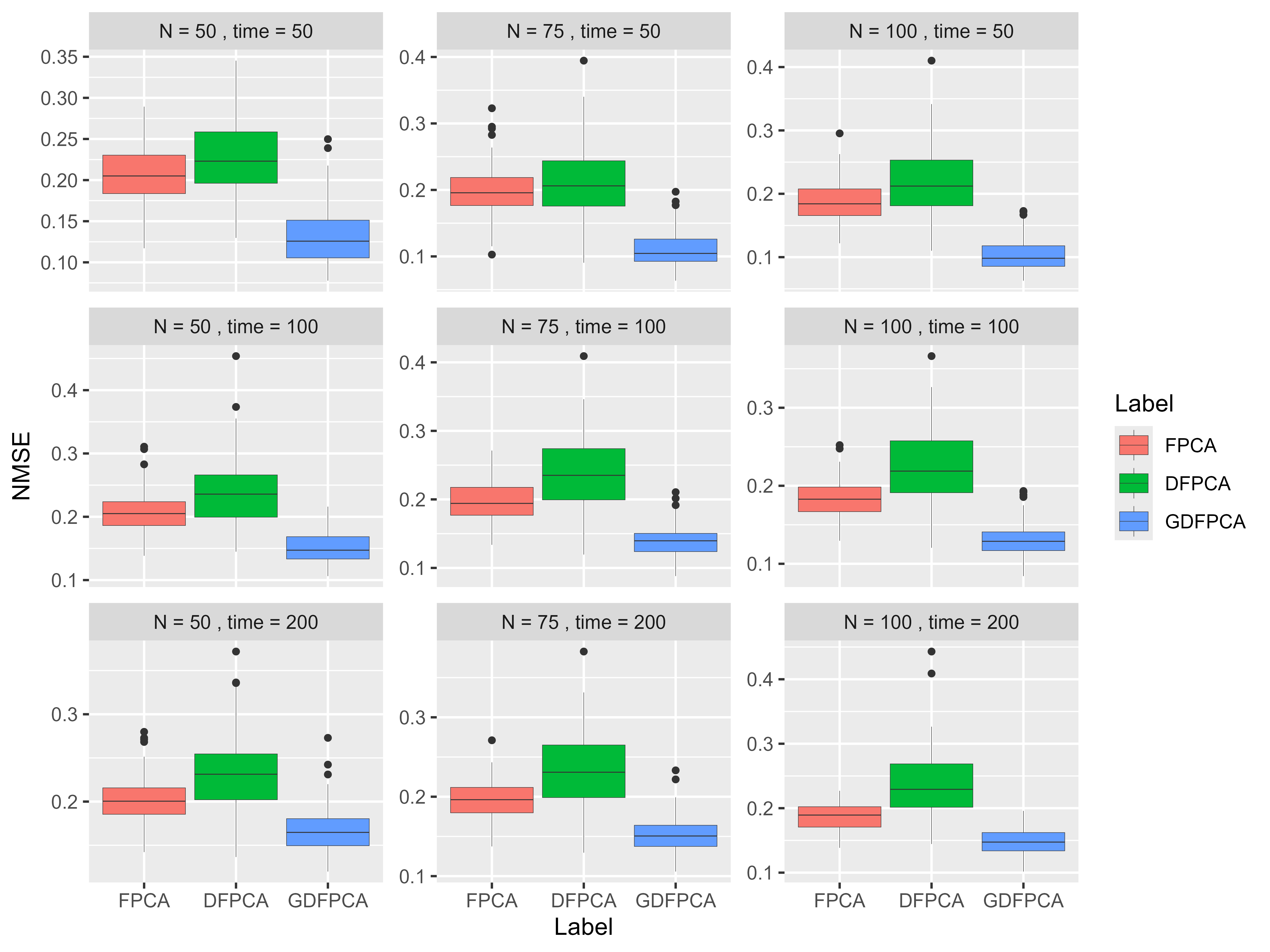}
                     \caption{Boxplots of NMSE over 100 Monte Carlo replications for simulated Wiener processes, comparing FPCA, DFPCA, and GDFPCA.}
                     \label{boxplot_NMSE_WIENER}
                \end{figure}

    \subsection{Simulation results for a dynamic factor model}

    In this section, we consider a dynamic factor model with a static factor representation \citep{forni2009opening}:
                \begin{align}
                   z_{it} &= \lambda_{i1}F_{1t} + \lambda_{i2}F_{2t} + ... +\lambda_{ir}F_{rt} + \epsilon_{it},\; \nonumber \\
                   \mathbf{F}_{t} &=\mathbf{D}\mathbf{F}_{t-1} + \mathbf{K}\mathbf{u}_{t}, \nonumber \\
                   i&=1,...,m,\; t=1,...,T, \nonumber 
                \end{align}
                where $\mathbf{F}_{t} = (F_{1t},F_{2t},...,F_{rt})^{'}$ and $u_{t} = (u_{1t},u_{2t},...,u_{qt})^{'}$, $\mathbf{D}$ is a matrix $r \times r$ and $K$ is a matrix $r \times q$, respectively. $u_{jt}$ and $\epsilon_{it}$ are standard normal random variables with i.i.d. $\lambda_{hi}$ and the entries of $\mathbf{K}$ are generated based on an independent and uniform distribution on the interval (-1,1), U(-1,1). Furthermore, the entries of $\mathbf{D}$ are generated according to section 4.1 in \citep{forni2017dynamic}. Firstly, each entry of $\mathbf{D}$ is generated based on an independent U(-1,1). Second, the resulting matrix from the first step is divided by its spectral norm to obtain the unit norm. Lastly, we multiply the resulting matrix by a random variable generated from U(-1,1). 
           
                The simulation outcomes show that choosing $p=6$ is enough to ensure that the explained variance reaches at least 80\% for all PCA frameworks. Figure \ref{boxplot_NMSE_FAC} presents the boxplots of the NMSE distributions for the reconstructions obtained from the three PCA frameworks. As the time of simulated series increases, GDFPCA surpasses the alternative methods by achieving the lowest NMSE.
                 \begin{figure}[!ht]
                     \centering
                     \includegraphics[width=0.5\textwidth]{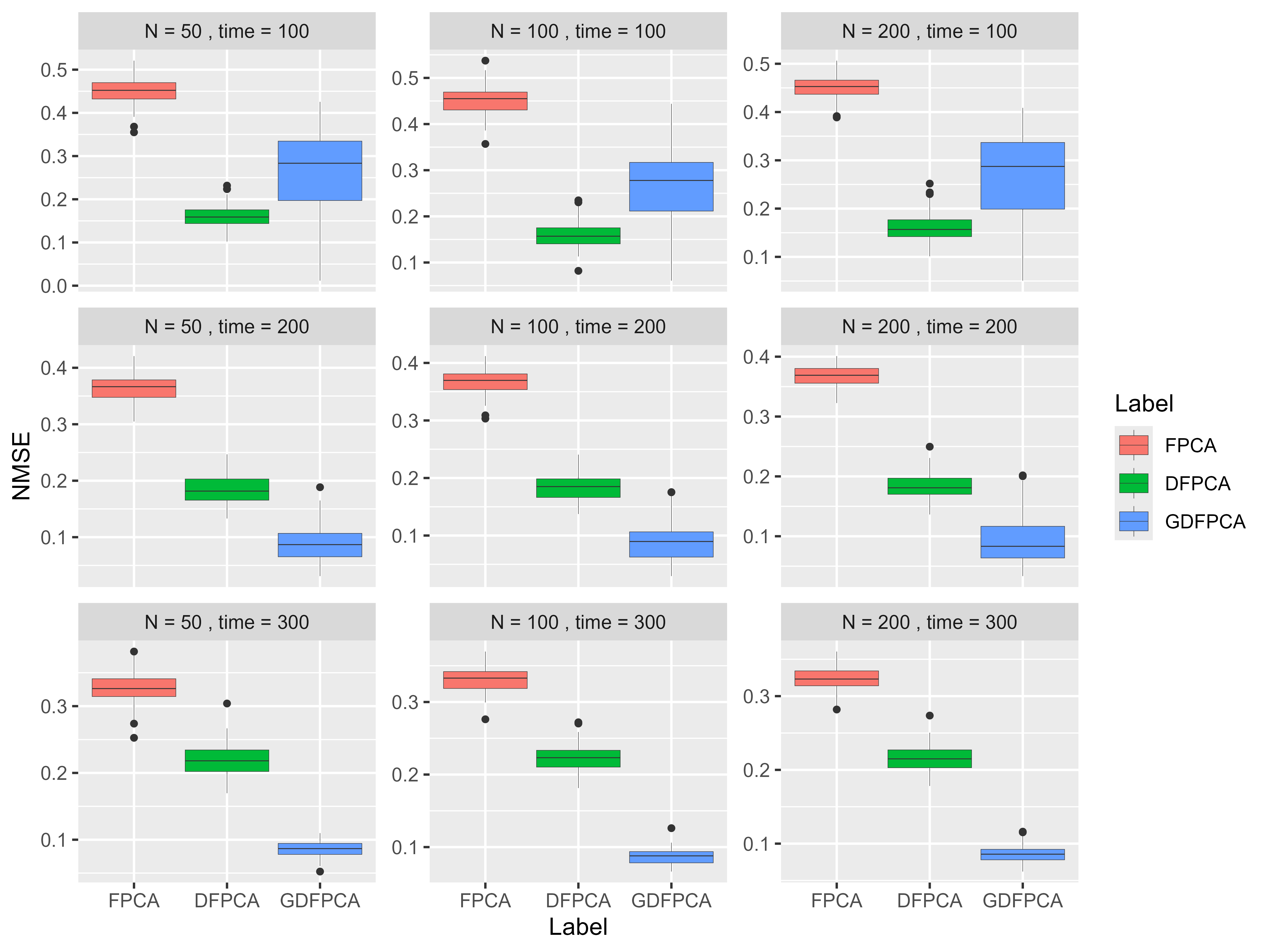} 
                     \caption{Boxplots of NMSE over 100 Monte Carlo replications for the dynamic factor model, comparing FPCA, DFPCA, and GDFPCA.}
                     \label{boxplot_NMSE_FAC}
                \end{figure}                       
                
\section{Applications to real data}
    \subsection{Particles with a diameter of 10 micrometers or less (\texttt{PM10})}
    In this case study, we apply GDFPCA to the \texttt{PM10} data set provided in the R package \texttt{freqdom.fda} \citep{freqdomfda}. The data consist of half-hourly measurements of particulate matter with an aerodynamic diameter below 10$\mu$m in ambient air, collected in Graz, Austria, over the period from October 1, 2010, to March 31, 2011. The original discrete observations have already been converted into functional data using 15 Fourier basis functions.    
    Figure \ref{pm10} displays the 175 daily trajectories from the \texttt{PM10} data set. In the same manner, we apply both FPCA and DFPCA to this data set. The proportions of variance explained by the three PCA-based methods are summarized in Table \ref{table_VAR_PM10}.
    
   We find that GDFPCA captures a proportion of explained variance comparable to that of DFPCA. Moreover, both GDFPCA and DFPCA approximation outperform those obtained via FPCA. This holds for all values of $p$ from 1 to 3.
   \begin{figure}[!ht]
         \centering
        \includegraphics[width=0.5\textwidth]{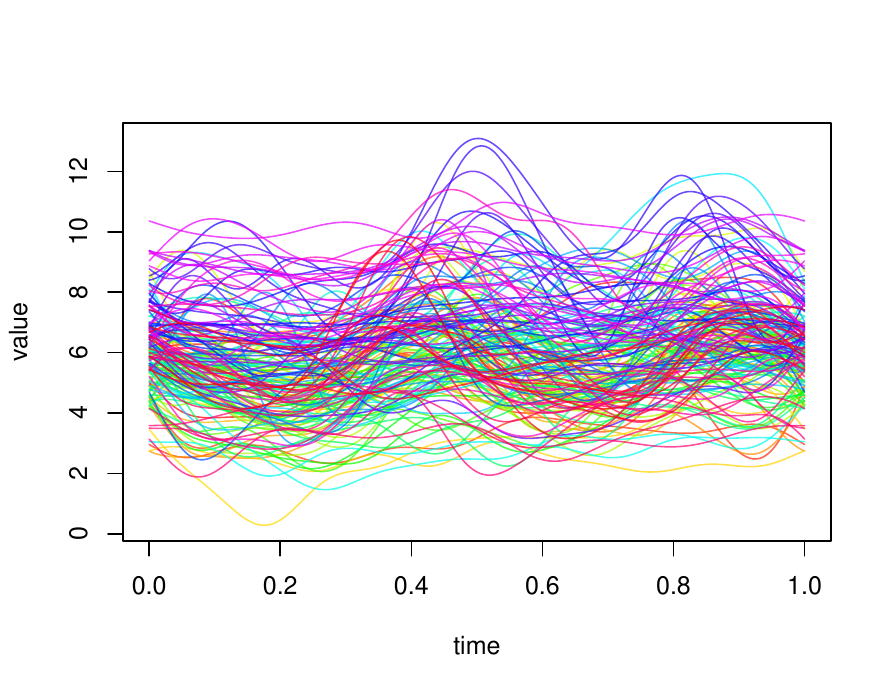}
       \caption{A plot of the data set PM10 in the form of functional time series.}
      \label{pm10}
   \end{figure}
    \begin{table*}[!ht]%
    \setlength\tabcolsep{0pt}
    \begin{tabular*}{\textwidth}{@{\extracolsep\fill}  l *5{d{1} @{\extracolsep\fill}}}
    \toprule
    Method  & \mc{PC1 (\%)} & \mc{PC2 (\%)} & \mc{PC3 (\%)} & \mc{Cumulative PC1-PC3 (\%)} & \mc{NMSE}
    \\ 
    \midrule
    FPCA                & \mc{72.3}   & \mc{10.29}  & \mc{6.1}  & \mc{88.7} &\mc{0.11} \\
    DFPCA               & \mc{77.9}   & \mc{7.3}  & \mc{5.1}   & \mc{90.4} &\mc{0.09}  \\
    GDFPCA              & \mc{81.1}   & \mc{8.5}  & \mc{4.1}  & \mc{93.8} &\mc{0.06} \\
    \bottomrule
    \end{tabular*}
    \captionsetup{singlelinecheck = false}
    \caption{Proportion of variance explained by the first three principal components, cumulative explained variance, and approximation NMSE for the PM10 daily functional data under FPCA, DFPCA, and GDFPCA.
}
    \label{table_VAR_PM10}
    \end{table*}

    \subsection{Standard \& Poor's 100 index (S\&P100)}
     We began by applying GDFPCA to 100 constituent stocks of the S\&P 100 as of September 2023. Daily closing prices from January 1, 2020, to January 1, 2021, were obtained from \url{http://www.investing.com}, and we computed the log-returns of these prices prior to conducting the GDFPCA. Because the stocks exhibit substantial variability, we first centered and scaled the data and then applied a smoothing procedure. GDFPCA was subsequently carried out on the smoothed series. For comparison, we also implemented standard FPCA and DFPCA on the same dataset. The proportions of variance explained by each of the three PCA frameworks are summarized in Table \ref{table_VAR_SP100}. Figure \ref{sp100} displays the 100 stock series represented as functional time series.
    We find that GDFPCA explains a proportion of the variance comparable to that of FPCA. Both GDFPCA and FPCA yield approximations that outperform those obtained with DFPCA, for all values of $p$ from 1 to 3. When $p=3$, GDFPCA achieves the highest proportion of explained variance. This outcome is anticipated, since the data are non-stationary, whereas DFPCA assumes stationarity \citep{hormann2015dynamic}.
   \begin{figure}[ht]
         \centering
        \includegraphics[width=0.5\textwidth]{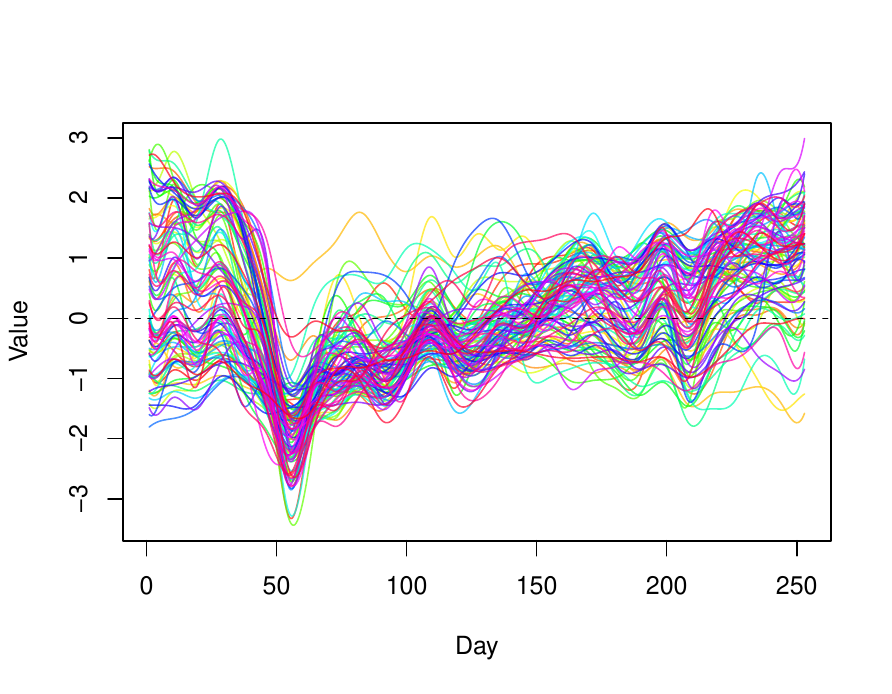}
       \caption{A plot of the centered, scaled and smoothed time series of the 100 stocks in S\&P100.}
      \label{sp100}
   \end{figure}
    \begin{figure}[h]
    \centering
         \includegraphics[width=0.5\textwidth]{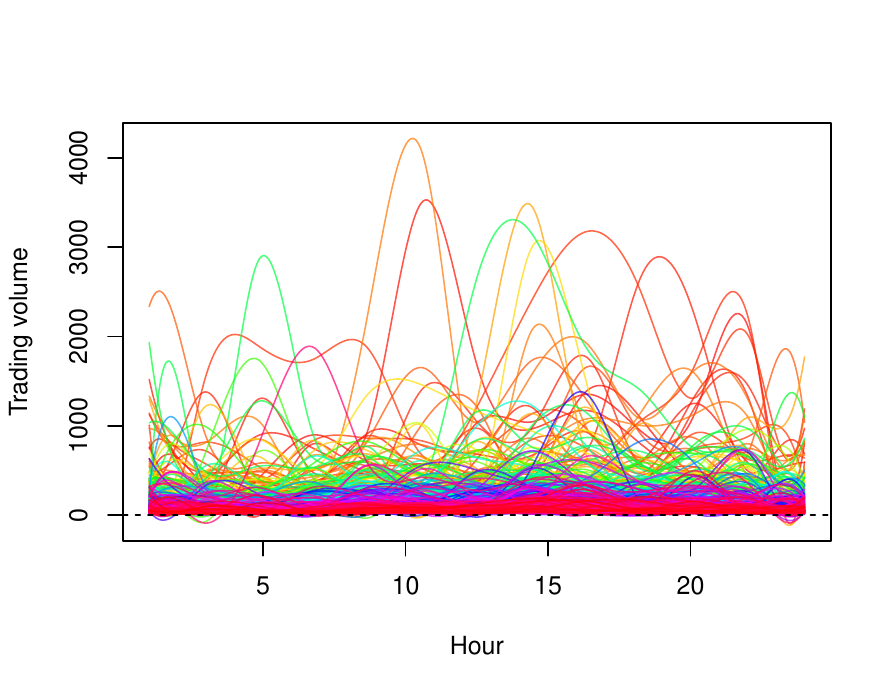}
         \caption{(a) A plot of the trading volumes of BTC in the form of functional time series.}
         \label{btc}
    \end{figure}
    \begin{table*}[ht]
    \setlength\tabcolsep{0pt}
    \begin{tabular*}{\textwidth}{@{\extracolsep\fill}  l *5{d{1} @{\extracolsep\fill}}}
    \toprule
    Method & \mc{PC1 (\%)} & \mc{PC2 (\%)} & \mc{PC3 (\%)} & \mc{Cumulative PC1-PC3 (\%)} & \mc{NMSE}
    \\ 
    \midrule
    FPCA                & \mc{51.5}      & \mc{14.7}   & \mc{7.1} & \mc{73.4} & \mc{0.266} \\
    DFPCA               & \mc{56.9}      & \mc{16.5}   & \mc{15.4} & \mc{81.6} & \mc{0.184} \\
    GDFPCA              & \mc{60.3}      & \mc{15.4}   & \mc{8.7} 
    & \mc{84.4} & \mc{0.156} \\
    \bottomrule
    \end{tabular*}
    \caption{Proportion of variance explained by the first three principal components, cumulative explained variance, and approximation NMSE for the S\&P100 stock data under FPCA, DFPCA, and GDFPCA.
    }\label{table_VAR_SP100}
    \end{table*}
    \begin{table*}[ht]%
    \setlength\tabcolsep{0pt}
    \begin{tabular*}{\textwidth}{@{\extracolsep\fill}  l *5{d{1} @{\extracolsep\fill}}}
    \toprule
    Method  & \mc{PC1 (\%)} & \mc{PC2 (\%)} & \mc{PC3 (\%)} & \mc{Cumulative PC1-PC3 (\%)} & \mc{NMSE}
    \\ 
    \midrule
    FPCA                & \mc{33.5}   & \mc{15.9}  & \mc{11.8}  & \mc{61.3} &\mc{0.38} \\
    DFPCA               & \mc{41.4}   & \mc{12.9}  & \mc{9.6}   & \mc{64.0} &\mc{0.35}  \\
    GDFPCA              & \mc{47.5}   & \mc{19.6}  & \mc{10.4}  & \mc{77.7} &\mc{0.22} \\
    \bottomrule
    \end{tabular*}
    \captionsetup{singlelinecheck = false}
    \caption{Proportion of variance explained by the first three principal components, cumulative explained variance, and approximation NMSE for the trading volume of Bitcoin under FPCA, DFPCA, and GDFPCA.
} \label{table_NMSE_BTC}
    \end{table*}
    \subsection{Trading Volume of Bitcoin}
    In this section, we evaluate and compare the performance of three PCA frameworks using Bitcoin trading volume data hourly from 2021. The dataset is obtained from \url{http://www.bitstamp.net}. It is organized as a 24 x 365 matrix, where each row corresponds to an hour of the day and each column corresponds to a calendar day. Figure \ref{btc} displays the smoothed Bitcoin trading volume over time. The proportions of variance captured by the three PCA frameworks are summarized in Table \ref{table_NMSE_BTC}.
    
    The findings show that DFPCA explains the largest proportion of variance when $p=1$. Nevertheless, as $p$ grows, GDFPCA yields greater gains in explained variance.

\section{Conclusions}
We introduced a functional extension of GDPCA \citep{pena2016generalized} as a data-driven approach to reduce dimensionality in functional time series. In this framework, we approximate the original functional time series by truncating the functional dynamic factor model.

The results from three simulation studies with different data-generating processes indicate that, under stationarity, GDFPCA and DFPCA achieve comparable approximation accuracy to the original functional time series, measured by explained variance and NMSE. In contrast, when the data-generating process is non-stationary, GDFPCA consistently outperforms both DFPCA and FPCA. This performance difference is theoretically consistent with the underlying assumptions of the competing methods. Spectral-based DFPCA relies on the existence of a well-behaved spectral density operator, which can become ill-behaved or singular at frequency zero in the presence of unit roots. By contrast, GDFPCA evaluates approximation error directly along the observed time path and therefore does not rely on strict stationarity or frequency-domain representations to extract dominant dynamic structure.

The applications to real datasets support the conclusions drawn from simulation studies. This indicates that GDFPCA is flexible  for dimensionality reduction across a wide range of functional data types, irrespective of whether the data are stationary or nonstationary.

The applicability of GDFPCA can be further broadened to improve the forecasting of both stationary and non-stationary functional time series models. Although most existing studies primarily address the stationary setting of functional time series \citep{aue2015prediction,shang2022dynamic,Chang2025}, it is anticipated that GDFPCA will also provide a useful representation for the analysis of non-stationary time series.

\backmatter

\bmhead{Data Availability}
Data will be made available upon reasonable request

\section*{Declarations}

\begin{itemize}
\item \textbf{Funding} This work was supported by the Universiti Malaya Research Excellence Grant (UMREG 2.0) 2024 [UMREG044-2024].

\item \textbf{Conflict of interest} The authors declare no potential conflict of interest

\item \textbf{Competing interests} The authors declare no competing interests,
\item \textbf{Author contribution} D.P. and S.D. conceived the study and supervised the project. I.M.D., T.H.K., S.D., and D.P. developed the theoretical framework. T.H.K., D.P., and S.D. developed the software. T.H.K. performed the simulations and wrote the original draft. D.P. acquired the funding. T.H.K. wrote the first draft of the manuscript. All authors reviewed and approved the final manuscript.
\end{itemize}

\begin{appendices}

\section{}\label{appendix}
Lemma 1 and Theorem 1 were given and proven by \cite{smucler2019consistency}. For the purpose of completeness, we outline the main steps of the proofs of Lemma 1 and Theorem 1 as follows: 
\\~\\
\noindent\textbf{Lemma 1:}
Assume A6 and A7 hold, then,
        $$
        \frac{\parallel \mathcal{E} \parallel}{\sqrt{nm}} = O_P\left(\frac{1}{m^{1/4}}\right)
        $$

\begin{proof}[Proof of Lemma 1]
 Consider $ v\in \mathbb R ^{m}$ with $\parallel v \parallel = 1$, then using Cauchy-Schwarz inequality we get: 
        $$
        \frac{v^T \mathcal{E}^T\mathcal{E}v}{nm} \leq \left( \frac{1}{m^2} \sum_{i,j = 1}^m \left( \frac{1}{n} \sum_{t = 1}^n \varepsilon_{t,i}\varepsilon_{t,j}\right)^2\right)^{1/2}.
        $$
        Thus we have:
        
        \begin{align}
            \left(\frac{\parallel \mathcal{E} \parallel^2}{nm}\right)^2 &= \left(\sup_{\parallel v \parallel = 1}\frac{v^T \mathcal{E}^T\mathcal{E}v}{nm} \right)^2 \nonumber\\ &\leq \frac{1}{m^2} \sum_{i,j = 1}^m \left( \frac{1}{n} \sum_{t = 1}^n \varepsilon_{t,i}\varepsilon_{t,j}\right)^2.
        \end{align}

        Using A6 and A7 we finally get: 
        
        $$
        \left(\frac{\parallel \mathcal{E} \parallel}{\sqrt{nm}}\right)^4 = O_P\left(\frac{1}{m}\right).
        $$    
\end{proof}
\begin{proof}[Proof of Theorem~{\upshape\ref{thm1}}]
 We have: 
        \begin{align}\label{eq:theorem1_inequality}
        &\dfrac{1}{nm}\sum_{t=1}^{n}\left\| \chi_{t}-\chi^{R,b}_{t}\FBetaAlphaMULhat\right\|^2 \nonumber\\
        &\leq 
        \dfrac{1}{nm}\sum_{t=1}^{n}\left\| \psi^m_{t} + \varepsilon_{t} -\chi^{R,b}_{t}(\hat{\mathbf{f}},\hat{\boldsymbol{\beta}^m},0) \right\|^2 \nonumber \\
        &=
        \dfrac{1}{nm}\sum_{t=1}^{n}\sum_{j=1}^m \varepsilon_{t,j}^2 \nonumber \\
        &=
        \dfrac{\parallel \mathcal{E}\parallel^2}{nm},    
        \end{align}
then, as $m = o(n)$ we get from Lemma 1:
        \begin{align}
        \dfrac{1}{nm}\sum_{t=1}^{n}\left\| \chi_{t}-\chi^{R,b}_{t}\FBetaAlphaMULhat \right\|^2 = O_{P}\left( \dfrac{1}{m^{1/2}} \right). \nonumber
        \end{align}
Notice that,
        \begin{align}\label{eq:theorem1_equation}
        &\dfrac{1}{nm}\sum_{t=1}^{n}\left\| \chi_{t}-\chi^{R,b}_{t}\FBetaAlphaMULhat \right\|^2 \nonumber \\ &= 
        \dfrac{1}{nm}\sum_{t=1}^{n}\sum_{j=1}^m \varepsilon_{t,j}^2 + \dfrac{1}{nm}\sum_{t=1}^{n}\left\| \psi^m_{t}-\chi^{R,b}_t\FBetaAlphaMULhat \right\|^2  \nonumber \\
        &- \frac{2\langle \mathcal{E}, \hat{\chi} - \Psi\rangle}{nm}.  
        \end{align}
By combining \eqref{eq:theorem1_inequality} and \eqref{eq:theorem1_equation}, we get: 
        $$
        \dfrac{1}{nm}\sum_{t=1}^{n}\left\| \psi^m_{t}-\chi^{R,b}_t\FBetaAlphaMULhat \right\|^2 \leq \frac{2\langle \mathcal{E}, \hat{\chi} - \Psi\rangle}{nm}.
        $$
Then, it is sufficient to prove that: 
        \begin{align}\label{eq:theorem1_proved_eq_1}
            \frac{2\langle \mathcal{E}, \hat{\chi} - \Psi\rangle}{nm} = O_{P}\left(\frac{1}{m^{1/4}}\right),
        \end{align}
Since $\chi^\top_t = (\beta^m)^\top F_t + \varepsilon_t^\top$, an application of triangular inequality on the corresponding matrices implies that $\left\| \chi \right\| \leq \left\| \beta^m \right\| \left\| F \right\| + \left\| \mathcal{E} \right\|$. Assumption A5 implies that $\left\| \beta^m \right\| = O(m^{1/2})$, and $F_{t}$ is second-order stationary, we have $\left\|  F \right\|$ = $O_{P}(n^{1/2})$. By Lemma 1, $\left\| \chi \right\|$ = $O_{P}((nm)^{1/2})$ and $\left\| \Psi \right\|$ = $O_{P}((nm)^{1/2})$. Since $\hat{\chi}$ is obtained by projecting $\chi$ on the space generated by the columns of $F_{t}$ and the vector with $n$ coordinates all equal to 1, $\left\| \hat{\chi} \right\|  \leq \left\| \chi \right\|$. Both $\Psi$ and $\hat{\chi}$ have rank bounded by $r = k+2$. Hence, it is enough to prove that for each fixed $C_1 > 0$, 
        $$
        \sup_{\parallel L \parallel \leq (nm)^{1/2}C_1, rank \leq r} \left| \dfrac{\langle \mathcal{E},L \rangle}{nm}\right|  = O_{P}\left(\dfrac{1}{m^{1/4}}\right).
        $$
We take $L$ with $\left\| L \right\| \leq (Tm)^{1/2}C_1$, and rank($L$) $\leq r$. Using the Singular Value Decomposition, we write $L = \sum_{l=1}^{r}\sigma_{l}u_{l}v_{l}^\top$, with $u_{l}$ and $v_{l}$ having unit norm. Then, 
    \begin{align}
        \left| \dfrac{\langle \mathcal{E},L \rangle}{nm}\right|
        &\leq 
        \dfrac{1}{nm}\sum_{l=1}^{r}\sigma_{l}\left| \langle \mathcal{E}, u_{l}v_{l} \rangle \right| \nonumber \\
        &\leq
        \dfrac{1}{nm}\max_{l \leq r}\sigma_{l}\sum_{l=1}^{r}\left| \langle \mathcal{E}, u_{l}v_{l} \rangle \right| 
        \nonumber\\
        &= \dfrac{1}{nm}\left\| L \right\|\sum_{l=1}^{r}\left| \langle \mathcal{E}, u_{l}v_{l} \rangle \right|. \nonumber
    \end{align}
Take any $u \in \mathbb{R}^T, v \in \mathbb{R}^m$ with unit norm. By using the cyclic property of the trace, $\left| \langle \mathcal{E}, u v^\top \rangle \right|$ = $\left| \text{tr}(u v^\top \mathcal{E}  ) \right|$ = $\left| \text{tr}(u^\top \mathcal{E} v ) \right| \leq \left\| \mathcal{E} \right\|$. With Lemma 1, we get:
        \begin{align}
            \left| \dfrac{\langle \mathcal{E}, L \rangle}{nm}\right|
            \leq
            \dfrac{r \left\| \mathcal{E} \right\| \left\| L \right\|}{nm} 
            \leq 
            \dfrac{r \left\| \mathcal{E} \right\| C_1}{(nm)^{1/2}} 
            = O_{P}\left(\dfrac{1}{m^{1/4}}\right). \nonumber
        \end{align}
Thus, by combining (\ref{eq:theorem1_inequality}), (\ref{eq:theorem1_equation}), (\ref{eq:theorem1_proved_eq_1}), and $m = o(n)$ we get:
        \begin{align}
        \dfrac{1}{nm}\sum_{t=1}^{n}\left\| \psi^m_{t}-\chi^{R,b}_{t}\FBetaAlphaMULhat  \right\|^2  = O_{P}\left(\dfrac{1}{m^{1/4}}\right).  
        \end{align}
\end{proof}

\begin{proof}[Proof of Theorem~{\upshape\ref{thm2}}]\label{proof_thm2}
Recall that $\{\phi_j\}$ are the eigenbasis used to expand the functional variables:
        $$
        X_t(u)=\sum_{j\ge 1} \chi_{t,j} \phi_j(u),\; X_{t}^{R,b}(u)= \sum_{j=1}^m \chi_{t,j}^{R,b} \phi_j(u).
        $$
Let $X^{R,b}_{t}\FBetaAlphahat$ be denoted here
        $X^{R,b}_{t}$ for simplicity of notation.
Then, we have,
        \begin{align}\label{eq:theorem_2_1st_equation}
          &\frac{1}{nm}\sum_{t=1}^{n}\left\| X_t-X^{R,b}_{t}\FBetaAlphahat \right\|^2_\mathcal{H}  \nonumber \\
          &= \frac{1}{nm}\sum_{t=1}^n (\chi_{t}-\chi_{t}^{R,b})^\top (\chi_{t}-\chi_{t}^{R,b}) + \frac{1}{nm}\sum_{t=1}^n \sum_{j > m} (\chi_{t,j})^2  
          \nonumber \\
          &= \frac{1}{nm} \sum_{t=1}^n \left\| \chi_{t}-\chi_{t}^{R,b} \right\|^2 +  O\left(\frac{1}{m^{2\tau}}\right)=O\left(\frac{1}{m^{1/2}}\right)
           \nonumber
        \end{align} 
This fulfils the proof by the help of Theorem {\upshape\ref{thm1}}.
\end{proof}

\end{appendices}


\bibliography{sn-bibliography}

@article{cai2006prediction,
  title={Prediction in functional linear regression},
journal={The Annals of Statistics},
  author={Cai, T Tony and Hall, Peter},
volume={34},
number={5},
  pages={2159--2179},
  year={2006}
}

@book{bosq2000linear,
  title     = {Linear Processes in Function Spaces: Theory and Applications},
  author    = {Bosq, Denis},
  series    = {Lecture Notes in Statistics},
  volume    = {149},
  year      = {2000},
  publisher = {Springer},
  address   = {New York}
}

@article{bosq2002estimation,
  title={Estimation of mean and covariance operator of autoregressive processes in Banach spaces},
  author={Bosq, Denis},
  journal={Statistical inference for stochastic processes},
  volume={5},
  number={3},
  pages={287--306},
  year={2002},
  publisher={Springer}
}

@article{mas2002,
title = {Weak convergence for the covariance operators of a Hilbertian linear process},
journal = {Stochastic Processes and their Applications},
volume = {99},
number = {1},
pages = {117-135},
year = {2002},
issn = {0304-4149},
doi = {https://doi.org/10.1016/S0304-4149(02)00087-X},
url = {https://www.sciencedirect.com/science/article/pii/S030441490200087X},
author = {André Mas}
}

@article{hormann2015dynamic,
  title={Dynamic functional principal components},
  author={H{\"o}rmann, Siegfried and Kidzi{\'n}ski, {\L}ukasz and Hallin, Marc},
  journal={Journal of the Royal Statistical Society Series B: Statistical Methodology},
  volume={77},
  number={2},
  pages={319--348},
  year={2015},
  publisher={Oxford University Press}
}

@article{hormann2010,
author = {Siegfried H{\"o}rmann and Piotr Kokoszka},
title = {{Weakly dependent functional data}},
volume = {38},
journal = {The Annals of Statistics},
number = {3},
publisher = {Institute of Mathematical Statistics},
pages = {1845 -- 1884},
year = {2010},
doi = {10.1214/09-AOS768},
URL = {https://doi.org/10.1214/09-AOS768}
}

@article{pena2016generalized,
  title={Generalized dynamic principal components},
  author={Pe{\~n}a, Daniel and Yohai, Victor J},
  journal={Journal of the American Statistical Association},
  volume={111},
  number={515},
  pages={1121--1131},
  year={2016},
  publisher={Taylor \& Francis}
}

@article{smucler2019consistency,
  title={Consistency of generalized dynamic principal components in dynamic factor models},
  author={Smucler, Ezequiel},
  journal={Statistics \& Probability Letters},
  volume={154},
  pages={108536},
  year={2019},
  publisher={Elsevier}
}

@article{forni2009opening,
  title={Opening the black box: Structural factor models with large cross sections},
  author={Forni, Mario and Giannone, Domenico and Lippi, Marco and Reichlin, Lucrezia},
  journal={Econometric Theory},
  volume={25},
  number={5},
  pages={1319--1347},
  year={2009},
  publisher={Cambridge University Press}
}

@article{forni2017dynamic,
  title={Dynamic factor models with infinite-dimensional factor space: Asymptotic analysis},
  author={Forni, Mario and Hallin, Marc and Lippi, Marco and Zaffaroni, Paolo},
  journal={Journal of Econometrics},
  volume={199},
  number={1},
  pages={74--92},
  year={2017},
  publisher={Elsevier}
}

@Manual{freqdomfda,
  title = {freqdom.fda: Functional Time Series: Dynamic Functional Principal Components},
  author = {Hormann, S. and Kidzinski, L.},
  year = {2022},
  note = {R package version 1.0.1},
  url = {https://CRAN.R-project.org/package=freqdom.fda},
}

@book{brillinger2001time,
  title     = {Time Series: Data Analysis and Theory},
  author    = {Brillinger, David R.},
  year      = {2001},
  publisher = {SIAM},
  address   = {Philadelphia}
}

@book{Ramsay2005-RAMFDA-3,
  title     = {Functional Data Analysis},
  author    = {Ramsay, James O. and Silverman, Bernard W.},
  edition   = {2},
  year      = {2005},
  publisher = {Springer},
  address   = {New York}
}

@article{gao2019high,
  title={High-dimensional functional time series forecasting: An application to age-specific mortality rates},
  author={Gao, Yuan and Shang, Han Lin and Yang, Yanrong},
  journal={Journal of Multivariate Analysis},
  volume={170},
  pages={232--243},
  year={2019},
  publisher={Elsevier}
}

@article{shang2022dynamic,
  title={Dynamic functional time-series forecasts of foreign exchange implied volatility surfaces},
  author={Shang, Han Lin and Kearney, Fearghal},
  journal={International Journal of Forecasting},
  volume={38},
  number={3},
  pages={1025--1049},
  year={2022},
  publisher={Elsevier}
}

@article{martinez2023surface,
  title={Surface time series models for large spatio-temporal datasets},
  author={Mart{\'\i}nez-Hern{\'a}ndez, Israel and Genton, Marc G},
  journal={Spatial Statistics},
  volume={53},
  pages={100718},
  year={2023},
  publisher={Elsevier}
}

@article{yang2022forecasting,
  title={Forecasting Australian fertility by age, region, and birthplace},
  author={Yang, Yang and Shang, Han Lin and Raymer, James},
  journal={International Journal of Forecasting},
  year={2022},
  publisher={Elsevier}
}

@article{aue2015prediction,
  title={On the prediction of stationary functional time series},
  author={Aue, Alexander and Norinho, Diogo Dubart and H{\"o}rmann, Siegfried},
  journal={Journal of the American Statistical Association},
  volume={110},
  number={509},
  pages={378--392},
  year={2015},
  publisher={Taylor \& Francis}
}

@article{liebl2013modeling,
  title={Modeling and forecasting electricity spot prices: A functional data perspective},
  author={Liebl, Dominik},
  journal={The Annals of Applied Statistics},
  pages={1562--1592},
  year={2013},
  publisher={JSTOR}
}

@article{martinez2022nonparametric,
  title={Nonparametric estimation of functional dynamic factor model},
  author={Mart{\'\i}nez-Hern{\'a}ndez, Israel and Gonzalo, Jes{\'u}s and Gonz{\'a}lez-Far{\'\i}as, Graciela},
  journal={Journal of Nonparametric Statistics},
  volume={34},
  number={4},
  pages={895--916},
  year={2022},
  publisher={Taylor \& Francis}
}

@article{hays2012functional,
  title={Functional dynamic factor models with application to yield curve forecasting},
  author={Hays, Spencer and Shen, Haipeng and Huang, Jianhua Z},
  journal={The Annals of Applied Statistics},
  pages={870--894},
  year={2012},
  publisher={JSTOR}
}

@article{giannone2006vars,
  title={VARs, common factors and the empirical validation of equilibrium business cycle models},
  author={Giannone, Domenico and Reichlin, Lucrezia and Sala, Luca},
  journal={Journal of Econometrics},
  volume={132},
  number={1},
  pages={257--279},
  year={2006},
  publisher={Elsevier}
}

@article{forni2000generalized,
  title={The generalized dynamic-factor model: Identification and estimation},
  author={Forni, Mario and Hallin, Marc and Lippi, Marco and Reichlin, Lucrezia},
  journal={Review of Economics and statistics},
  volume={82},
  number={4},
  pages={540--554},
  year={2000},
  publisher={MIT Press 238 Main St., Suite 500, Cambridge, MA 02142-1046, USA journals~…}
}

@article{stock1988testing,
  title={Testing for common trends},
  author={Stock, James H and Watson, Mark W},
  journal={Journal of the American statistical Association},
  volume={83},
  number={404},
  pages={1097--1107},
  year={1988},
  publisher={Taylor \& Francis}
}

@article{bai2002determining,
  title={Determining the number of factors in approximate factor models},
  author={Bai, Jushan and Ng, Serena},
  journal={Econometrica},
  volume={70},
  number={1},
  pages={191--221},
  year={2002},
  publisher={Wiley Online Library}
}

@article{forni2005generalized,
  title={The generalized dynamic factor model: one-sided estimation and forecasting},
  author={Forni, Mario and Hallin, Marc and Lippi, Marco and Reichlin, Lucrezia},
  journal={Journal of the American statistical association},
  volume={100},
  number={471},
  pages={830--840},
  year={2005},
  publisher={Taylor \& Francis}
}

@article{lam2011estimation,
  title={Estimation of latent factors for high-dimensional time series},
  author={Lam, Clifford and Yao, Qiwei and Bathia, Neil},
  journal={Biometrika},
  volume={98},
  number={4},
  pages={901--918},
  year={2011},
  publisher={Oxford University Press}
}

@article{bai2004panic,
  title={A PANIC attack on unit roots and cointegration},
  author={Bai, Jushan and Ng, Serena},
  journal={Econometrica},
  volume={72},
  number={4},
  pages={1127--1177},
  year={2004},
  publisher={Wiley Online Library}
}

@article{pena2006nonstationary,
  title={Nonstationary dynamic factor analysis},
  author={Pe{\~n}a, Daniel and Poncela, Pilar},
  journal={Journal of Statistical Planning and Inference},
  volume={136},
  number={4},
  pages={1237--1257},
  year={2006},
  publisher={Elsevier}
}

@article{diebold2006forecasting,
  title={Forecasting the term structure of government bond yields},
  author={Diebold, Francis X and Li, Canlin},
  journal={Journal of econometrics},
  volume={130},
  number={2},
  pages={337--364},
  year={2006},
  publisher={Elsevier}
}

@article{kokoszka2015functional,
  title={Functional dynamic factor model for intraday price curves},
  author={Kokoszka, Piotr and Miao, Hong and Zhang, Xi},
  journal={Journal of Financial Econometrics},
  volume={13},
  number={2},
  pages={456--477},
  year={2015},
  publisher={Oxford University Press}
}

@book{horvath2012inference,
  title     = {Inference for Functional Data with Applications},
  author    = {Horv{\'a}th, Lajos and Kokoszka, Piotr},
  year      = {2012},
  publisher = {Springer}
}

@article{horvathfts2012,
    author = {Horváth, Lajos and Kokoszka, Piotr and Reeder, Ron},
    title = {Estimation of the Mean of Functional Time Series and a Two-Sample Problem},
    journal = {Journal of the Royal Statistical Society Series B: Statistical Methodology},
    volume = {75},
    number = {1},
    pages = {103-122},
    year = {2012},
    month = {06},
    issn = {1369-7412},
    doi = {10.1111/j.1467-9868.2012.01032.x},
    url = {https://doi.org/10.1111/j.1467-9868.2012.01032.x}
}

@article{panaretos:tavakoli:2013AOS,
  title = {{\GG{20130401}}{F}ourier Analysis of Stationary Time Series in Function Space.},
  author = {Victor M. Panaretos and Shahin Tavakoli},
  year = {2013},
  journal = {The Annals of Statistics},
  number = {2},
  pages = {568--603},
  volume = {41},
  date = {2013-01-01}
}

@article{panaretos:tavakoli:2013SPA,
  title = {{\GG{20130701}}{C}ram{\'e}r--{K}arhunen--{L}o{\`e}ve Representation and Harmonic Principal Component Analysis of Functional Time Series},
  author = {Victor M. Panaretos and Shahin Tavakoli},
  year = {2013},
  journal = {Stochastic Processes and their Applications},
  number = {7},
  pages = {2779--2807},
  volume = {123},
  date = {2013-01-01},
}

@article{forni2015dynamic,
  title={Dynamic factor models with infinite-dimensional factor spaces: One-sided representations},
  author={Forni, Mario and Hallin, Marc and Lippi, Marco and Zaffaroni, Paolo},
  journal={Journal of econometrics},
  volume={185},
  number={2},
  pages={359--371},
  year={2015},
  publisher={Elsevier}
}

@article{stock2002forecasting,
  title={Forecasting using principal components from a large number of predictors},
  author={Stock, James H and Watson, Mark W},
  journal={Journal of the American statistical association},
  volume={97},
  number={460},
  pages={1167--1179},
  year={2002},
  publisher={Taylor \& Francis}
}

@article{vandelft_dette_2024,
  author  = {van Delft, Anne and Dette, Holger},
  title   = {A General Framework to Quantify Deviations from Structural Assumptions in the Analysis of Nonstationary Function-Valued Processes},
  journal = {The Annals of Statistics},
  year    = {2024},
  volume  = {52},
  number  = {2},
  pages   = {550--579},
  doi     = {10.1214/24-AOS2358}
}

@article{Chang2025,
author = {Jinyuan Chang and Qin Fang and Xinghao Qiao and Qiwei Yao},
title = {On the Modeling and Prediction of High-Dimensional Functional Time Series},
journal = {Journal of the American Statistical Association},
volume = {120},
number = {552},
pages = {2181--2195},
year = {2025},
publisher = {Taylor \& Francis},
doi = {10.1080/01621459.2024.2413201}

}

@book{seber2003linear,
  title     = {Linear Regression Analysis},
  author    = {Seber, George A. F. and Lee, Alan J.},
  edition   = {2},
  year      = {2003},
  publisher = {John Wiley \& Sons, Inc.},
  address   = {Hoboken, NJ},
  doi       = {10.1002/9780471722199}
}

\end{document}